\documentclass[sigconf]{acmart}
\usepackage{balance}
\usepackage{amsmath}
\usepackage{amsthm}
\usepackage{amsfonts}
\usepackage{amssymb}
\usepackage{graphicx}
\usepackage[tight,footnotesize]{subfigure}

\usepackage{color}    
\usepackage{url}
\usepackage{mathrsfs}
\usepackage{msc}
\usepackage{pgf-pie}
\usepackage{threeparttable}
\usepackage{booktabs} 
\usepackage{multirow}
\usepackage{tabularx} 
\newcolumntype{L}[1]{>{\raggedright\arraybackslash}p{#1}}
\newcolumntype{C}[1]{>{\centering\arraybackslash}p{#1}}
\newcolumntype{R}[1]{>{\raggedleft\arraybackslash}p{#1}}
\usepackage{pgfplots}
\usepackage{scalefnt}
\usepackage[linesnumbered,commentsnumbered,ruled,vlined]{algorithm2e}
\usepackage{paralist}
\usepackage{tikz}
\pgfdeclarelayer{bg}
\pgfsetlayers{bg,main}
\usetikzlibrary{decorations.text} 
\usetikzlibrary{backgrounds}
\usetikzlibrary{decorations.pathreplacing}
\usepackage{listings}
\usepackage{pifont}
\newcommand{\xmark}{\ding{55}}%

\newtheorem{definition}{Definition}

\newtheorem{theorem}{Observation}

\makeatletter\def\@captype{table}\makeatother

\AtBeginDocument{%
  \providecommand\BibTeX{{%
    \normalfont B\kern-0.5em{\scshape i\kern-0.25em b}\kern-0.8em\TeX}}}

\makeatletter
\renewcommand\footnotetextcopyrightpermission[1]{}
\makeatother

\begin{document}
\fancyhead{}
\title{Fail-safe Watchtowers and Short-lived Assertions for Payment Channels}

\author{Bowen Liu}
\affiliation{%
  \institution{Singapore University of\\Technology and Design}
  \country{Singapore}
}
\email{bowen\_liu@mymail.sutd.edu.sg}

\author{Pawel Szalachowski}
\affiliation{%
  \institution{Singapore University of\\Technology and Design}
  \country{Singapore}
}
\email{pawel@sutd.edu.sg}

\author{Siwei Sun}
\affiliation{%
  \institution{State Key Laboratory of Information Security, Institute of Information Engineering, Chinese Academy of Sciences}
}
\affiliation{%
  \institution{School of Cyber Security, University of Chinese Academy of Sciences}
  \city{Beijing}
  \country{China}
}
\email{sunsiwei@iie.ac.cn}

\begin{abstract}
    The recent development of payment channels and their extensions (e.g., state channels)
    provides a promising scalability solution for
    blockchains which allows untrusting parties to transact off-chain and resolve
    potential disputes via on-chain smart contracts. 
    To protect participants who have no constant access to the blockchain,
    a watching service named as watchtower is proposed -- 
    a third-party entity obligated to
    monitor channel states (on behalf of the participants) and correct them
    on-chain if necessary. Unfortunately, currently proposed watchtower schemes
    suffer from multiple security and efficiency drawbacks.
    
    In this paper, we explore the design space behind watchtowers. 
    We propose a novel watching service named as \textit{fail-safe watchtowers}. 
    In contrast to prior proposed watching services, our fail-safe watchtower
    does not watch on-chain smart contracts constantly.
    Instead, it only sends a single on-chain message periodically confirming or
    denying the final states of channels being closed.
    Our watchtowers can easily handle a large number of channels, are
    privacy-preserving, and fail-safe tolerating multiple attack vectors. 
    Furthermore, we show that watchtowers (in general) may be an option economically
    unjustified for multiple payment scenarios and we introduce a simple, yet
    powerful concept of short-lived assertions which can mitigate misbehaving
    parties in these scenarios.
\end{abstract}



\maketitle

\section{Introduction}
\label{sec:intro}
Scalability is seen as the main limitation of currently deployed permissionless
blockchain systems~\cite{croman2016scaling}. 
Mainstream platforms like Bitcoin or Ethereum can
handle from a few to several transactions per second, 
and clearly, without
handling high transaction volumes these systems are unlikely to achieve mass
adoption.
While more performant consensus algorithms are being constantly proposed (i.e., layer-1
solutions)~\cite{algorand17,ouroboros17,snowwhite16,thunderella18,sleepy17}, 
the concept of payment channels has emerged in parallel.
Payment channels and network build upon them, like Lightning Network~\cite{lightening} or
Raiden~\cite{Raiden}, are layer-2 solutions which involve on-chain smart contracts only for
channel management (i.e., creation, termination, or dispute) while all regular
transactions are exchanged off-chain. 
That allows untrusting parties to transact off-chain as simple and fast as
exchanging signed messages which is especially important for micropayments.

A party can close a channel by sending the latest statement to the smart
contract, which verifies whether the statement is properly signed by 
both parties and transfers the tokens accordingly.  
The smart contract does not know what is the
off-chain state, thus the payout operation is timeouted to give another party
time to dispute. 
Disputes are resolved by the contract itself, simply accepting
the freshest message (e.g., messages can be ordered by unique ascending
counters).  
One security assumption in this setting is that a party has to be
online to detect and prevent misbehavior of its peer sending a stale message to
close the channel. 

To relax the \textit{always online} assumption, the concept of
outsourcing arbitration to so-called watchtowers (or monitors) was proposed~\cite{monitor16}.  
They are
highly available third parties which are contacted by transacting parties with
every transaction. 
Therefore, an involved watchtower knows the current off-chain
payment state and its duty is to monitor the on-chain state. 
In the case of a
misbehavior, i.e., when a party tries to close the channel using a stale state,
the watchtower will be able to trigger the dispute process by providing the
current state on behalf of the inactive party. 

Unfortunately, watchtower schemes present in the literature have some important
limitations (see details in \autoref{sec:related}). 
They need to observe every
channel contract separately, and they are not fail-safe, allowing misbehavior under
their unavailability, or still requiring a significant payout
timeout.\footnote{Long payout timeouts are undesired as 
a party's deposit cannot be used during this time.}

To mitigate these drawbacks we propose a concept of fail-safe watchtowers. 
In our construction, watchtowers do not watch channel contracts constantly anymore. 
Instead,
they observe and assert off-chain transactions, and periodically send a 
message to the blockchain, encoding the latest states of watched payment
channels.  
A party closing its channel has to authorize this operation by
confirming its state with the fresh information submitted by the watchtower.
Sending a ``positive'' information can eliminate long timeouts present in
current systems where watchtowers send only ``negative'' information (correcting
an incorrect state).
Moreover, our design of watchtowers incentivizes them to confirm (or reject) a given state as
soon as possible.
This construction allows to implement quick channel termination in the common
case (i.e., when the watchtower is online), is fail-safe as a longer timeout
will be triggered only when the watchtower is offline, and allows watchtowers to
scale better as they can be implemented as light blockchain clients.
Furthermore,
we show that for a large number of payment channels, 
watchtower schemes in general may be
economically ineffective. 
To improve security in these cases we introduce
short-lived assertions, that allow channel contracts to distinguish fresh and
stale assertions and process them accordingly (minimizing timeouts and
misbehavior impact). 

In this work we make the following contributions:
\begin{enumerate}
    \item we propose a novel construction of watchtowers, which in contrast to
        prior proposals, minimize payout delays in the common case (i.e., when
        the parties are honest), are fail-safe, efficient, and
        privacy-preserving (see \autoref{sec:watchtower}), 

    \item we show the economic limitations of watchtower schemes and introduce
    novel short-lived assertions for payment channels, which is a complementary
    solution suitable for multiple scenarios, like micropayments 
    (see \autoref{sec:assertions}),

    \item we implement and evaluate the proposed systems to demonstrate
    their security, feasibility, and efficiency  (see \autoref{sec:analysis}).
\end{enumerate}

\section{Motivation and Related Work }
\label{sec:related}
\paragraph{\textbf{Payment Channels.}} 
Currently existing decentralized permissionless blockchain platforms
(e.g., Bitcoin and Ethereum) can hardly compete against those centralized 
financial service providers like VISA or Paypal in terms of processed transactions per second (TPS).
To put this into perspective, Bitcoin only supports 
up to 7 TPS~\cite{croman2016scaling}, while VISA can handle an average of 150
million transactions every day or around 1736 TPS~\cite{visamainpage}.
Today, the scalability problem of such blockchains has been regarded as a
main obstacle for their further adoption as large scale payment networks.

The community has put forth various proposals to increase 
the throughput of blockchain platforms. 
Reparameterization of block sizes and block intervals of existing systems has been proved to be
ineffective~\cite{gervais2016security} in boosting the performance. 
Designing new consensus protocols
from scratch~\cite{luu2015scp,eyal2016bitcoin} is promising, but requires significant implementation, deployment, and adoption costs to be used in practice.
Moreover, it takes a long time to understand and build up the security confidence of new consensus protocols.  

Another line of research is to minimize the interactions with blockchains by using
off-chain protocols --- the so-called \textit{layer-2} solutions. 
Recently developed payment channels and state channels together with their extensions~\cite{bolt,chiesa2017decentralized,malavolta2017concurrency} belong to this approach.
The main idea of layer-2 solutions is to make payments and state transitions 
by exchanging off-chain messages among small groups of parties privately, and use the
blockchain as a backstop only when disputes arise.   
As a result, parties
can make a huge number of \textit{off-chain} transactions, while only 
their initial deposit and the finalized balances will be included \textit{on-chain}.
Moreover, this approach is backward compatible with existing systems and
is orthogonal to other scalability solutions, thus it has become 
a promising approach for improving the scalability of blockchains and cryptocurrencies.

The \textit{layer-2}
payment protocols have evolved from 
simplex and probabilistic 
payment channels~\cite{hearn_micpayment,Dmitrienko_simplex,pass_probchannel,hu_probchannel}
to fully duplex payment channels~\cite{decker2015fast,bolt,lind2017teechain,dziembowski2017perun}, 
based on which payment networks~\cite{lightening,Raiden,miller2017sprites,khalil2017revive} 
can be constructed where users can send payment to each other 
without establishing a direct channel from scratch. 
Besides, state channels~\cite{dziembowski2018general} generalize the concept of payment channels 
to support not only payments but also the execution of arbitrary state transitions.
We refer the readers to~\cite{sok1,sok2} for a more comprehensive discussion
of the layer-2 off-chain protocols.
In this work, we mainly focus on enhancing the security of payment channels, although our schemes can be generalized to state channels. 

\paragraph{\textbf{Watching Service.}} 
\label{sec:watching}
The dispute resolution mechanisms of payment channels require that the parties are
online and synchronized with the blockchain to prevent the channels from finalizing with invalid stale states.   
To alleviate this strong {\it always-online} requirement, 
schemes outsourcing the responsibility of monitoring the blockchains and issuing
challenges in case of disputes to third-party watching services are proposed. 
Users outsource their latest off-chain states to the watching service before parting offline. 
Watching services then act on behalf of the users to protect their funds. 
Users can still verify the correct behavior of watching services and punish them in case of ineligibility and non-compliance~\cite{sok2}. 

Monitor~\cite{monitor16} is a watching service for the Lightning network.
However, due to the nature of replace-by-revocation channels in the Lightning
network, the Monitor has to store all transactions within the channel, making it
inefficient to provide service to a large number of busy channels simultaneously.
The WatchTower scheme proposed by Osuntokun~\cite{hard_lightning15} reduces the storage
requirement of Monitor from $O(N)$ to $O(1)$. 
But this improvement relies on a new opcode to parse and verify a signed
message, which is currently unavailable in Bitcoin. 
Therefore, WatchTower cannot take the place of Monitor without updating the underlying layer-1 protocol.
More importantly, both Monitor and WatchTower lack a mechanism for customers to
recourse when the watching service fails to settle a dispute, since no verifiable evidence is provided.

To address the inefficiencies and limitations of prior watching services, 
PISA is proposed which works for generic state channels~\cite{mccorry2018pisa}. 
In PISA, the third party (named as a custodian)
appointed by customers to watch the state channel and cancel execution forks 
only needs to store the hash of the
most recent state, and thus only $O(1)$ storage is required.
Moreover, PISA provides publicly verifiable cryptographic evidence in case
the custodian fails, based on which the involved customer is able to penalize
the custodian.
However, the customer of custodian is implicit to the public and the custodian can 
use the same deposit as a safety for all customers. 

The always-online assumption is moved in PISA from parties to a
custodian and moreover the system is not designed in a fail-safe way -- i.e.,
an unavailable/failed custodian silently accepts malicious disputes, thus to
mitigate it
payout timeouts should be set long enough, freezing deposits of parties. 

The {\it Disclose Cascade Watch Commit} (DCWC) scheme~\cite{dcwc18} 
proposes to incentivize multiple third parties to provide watching services,
where only watchtowers whose participating evidences are included
on chain get paid. 
However, each party has to be responsive and synchronized with the 
blockchain at all times.
While DCWC relates to node failures as well, 
security in general and topology-based attacks in particular have been mostly neglected, though.
\medskip

\paragraph{\textbf{Watching Services and Trust.}} 
\label{sec:watching:trust}
Although it may be counter-intuitive that parties of a decentralized platform
employ a centralized entity (i.e., a watchtower) for dispute resolution,
there are essential differences between these approaches and the centralization
of the underlying blockchain platform.
First of all, watching services (as well as payment channels and other
upper-layer applications) are external to the platform and its consensus rules,
therefore they do not violate the decentralization of the platform itself.
Second, the trust placed in watching services is ``local'' to the parties that
deploy it, i.e., no other blockchain participants have to trust (or even be
aware of) such a service.  
Moreover, some designs (including ours) allow a client to
be represented by a watching service in a way opaque even to the peer of the
payment channel. 
Finally, designs of watching services can benefit from the properties of the
underlying platforms, which include transparency of service actions (allowing
early detection of misbehavior) or cryptoeconomic incentives aiming to
facilitate correct service behavior by punishments and rewards.

\section{Preliminaries}
\label{sec:pre}
\subsection{Blockchain and Smart Contract}
\label{sec:preblockchain}
Over the last decade, the blockchain technology has been pioneered by Bitcoin~\cite{crosby2016blockchain}
which enables mutually untrusted parties to reach consensus over the state of a distributed 
and decentralized global ledger confirming and serializing ``transactions''.
By using a similar consensus protocol and extending the scripting language of
Bitcoin to a Turing complete programming language, 
Ethereum~\cite{buterin2014ethereum} was built -- 
a general-purpose blockchain system
with its native cryptocurrency {\it ether} which can execute
programs over a decentralized and
replicated state machine named as the {\it Ethereum Virtual Machine}~(EVM). 
Due to the Turing completeness, with Ethereum
one can implement self-enforcing programs called smart contracts with 
nearly arbitrary business logic.

To prevent parties from intentionally or unintentionally abusing the system resources,
computation and memory utilization in Ethereum are charged
in {\it gas} which can be purchased with ether with certain exchange rate.
Therefore, maintaining a smart contract storing data objects of large size
and performing complex computations can be very expensive. 
To facilitate lightweight communication between smart contracts and deploying them clients, 
a mechanism of events and logs is introduced in Ethereum. 
Smart contracts can emit events and write logs to the blockchain. 
When they are called, the emitted messages
are stored in the transaction's log associated with the address of the contract.
Log and event data are not accessible from within contracts but can be captured
by clients reading (part of) the blockchain.

\subsection{Payment Channels}
A payment channel is a relationship established between two participants
that wish to transact, such that bi-directional payments
can be made by exchanging off-chain messages if both parties behave honestly.
The state of the relationship resulted from
those off-chain transactions
will eventually be recorded on blockchain, which only happens when
the payment channel is finalized. 
A set of rules are enforced by the 
blockchain to guarantee that the channel is finalized with the correct state even
if all parties are malicious. The life cycle of a payment channel 
can be divided into the following phases.

\paragraph{\textbf{Setup.}}
Before the payment channel is established, the involved parties need 
to deposit certain amount of coins to a payment channel contract that both parties
agree on. These coins are going to be redistributed continuously among them
during the payment phase.

\paragraph{\textbf{Payment.}}
During this phase the participants of the channel can make an arbitrary number of payments
(within the channel capacity induced by the initial deposit)
to each other by exchanging jointly signed off-chain messages.
Since there is no interaction with the blockchain, the transaction throughput
is only limited by computing resources and the actual communication network for exchanging messages between the parties, effectively bypassing the bottleneck due to the underlying consensus mechanism
of the blockchain. 
 
\paragraph{\textbf{Closure.}}
Any party can try to close the payment channel by issuing a state signed by both parties
to the blockchain. 
Before the channel is finalized with this submitted state,
the payment channel reserves a time window for the other party to intervene
and invalidate the previously submitted state with a {\it newer} signed state.
Eventually, the most recent state seen by the blockchain will be accepted after
a timeout.
According to this rule, the payment channel may end up with a non-latest state
against the legitimate interest of one side who fails to provide proper evidence to 
the payment channel contract
showing the invalidity of the previously submitted state within the time window. 
Therefore, it is essential to observe the blockchain constantly to ensure
that only the latest state will be accepted by the blockchain.

\section{Fail-safe Watchtowers}
\label{sec:watchtower}
\subsection{Overview}
\label{sec:overview}
Before proceeding with the details, we give a brief overview of our design,
which is depicted in~\autoref{fig:overview}.
\begin{figure}[h!]
	\centering
	\includegraphics[width=\linewidth]{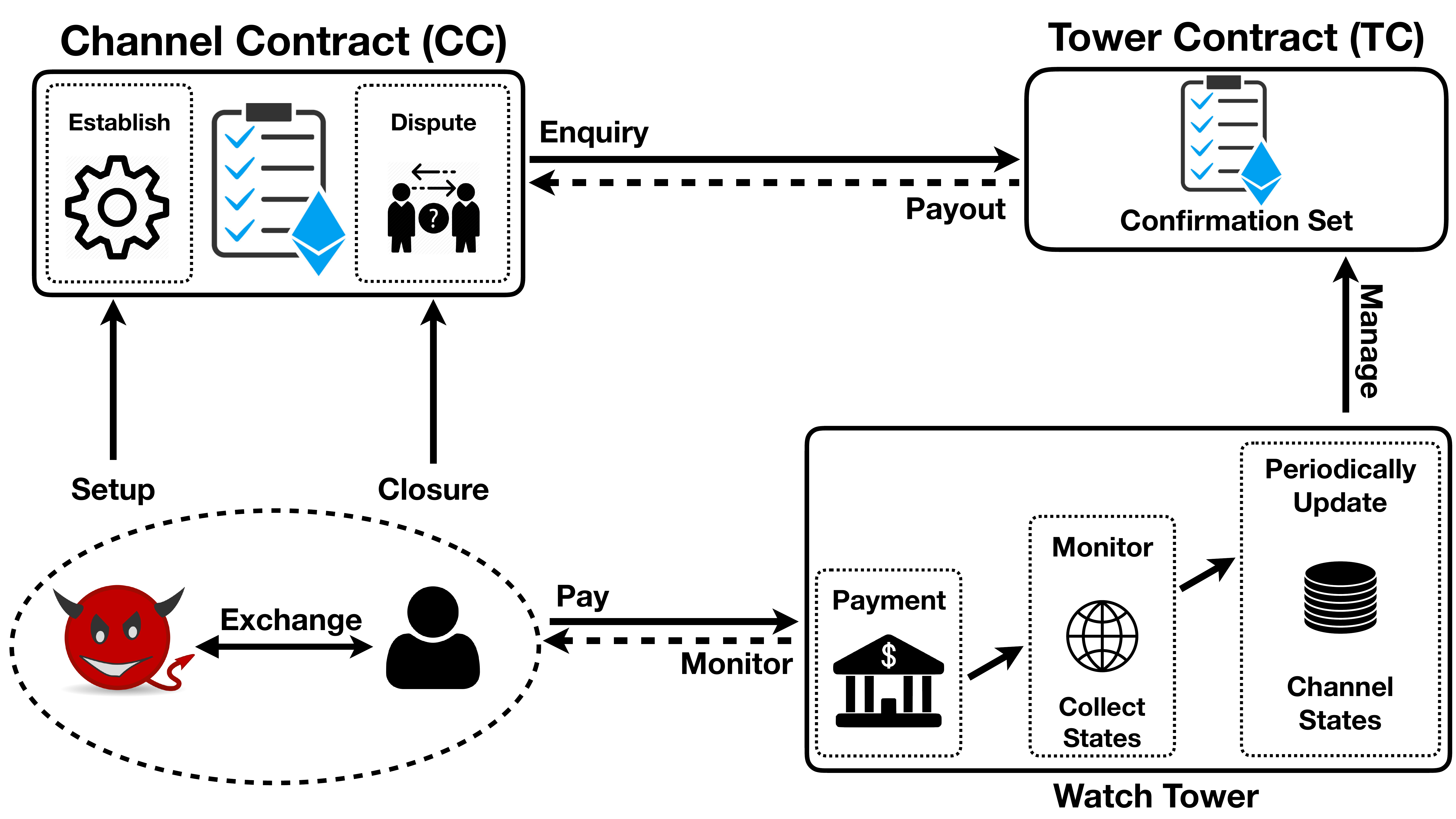}
    \caption{Protocol Overview.}
	\label{fig:overview}
\end{figure}

Two participants first register with a payment \textit{channel contract} by
depositing some initial funds there. 
Then, through the \textit{tower contract} they can employ
a third party called a watchtower to protect the payment state of their own channel.
Afterwards, the registered channel participants communicate off-chain
to authorize and make payments with each other, redistributing the
initial deposits amongst themselves.
To keep track of the transactions, a monotonic transaction counter is incremented for 
every transaction appearing in the channel. 
The watchtower receives and verifies every off-chain transaction, and records
the latest state of the payment channel.
A watchtower can monitor the off-chain states of multiple channels simultaneously. 
It manages a tower contract and periodically updates it with the 
so-called {\it confirmation set} which is used to determine payout process.
To close the channel, a party submits a closure transaction to the channel
contract which delegates the final payout decision to the tower contract. 
Any participant is able to
invoke dispute if any disagreement happens.

We design protocol with the following goals in mind:
\paragraph{\textbf{Security.}}
We assume an adversary who cannot compromise 
\begin{compactitem}
    \item the underlying cryptographic primitives (e.g., digital signatures,
        hash functions, etc.),
    \item the properties of the underlying blockchain platform, 
    \item the smart contract runtime environment, 
    \item a watchtower (but an adversary can be one peer of the payment
        channel); however, we discuss some cases of misbehaving watchtowers.  
\end{compactitem}

Under these assumptions our watchtower
protocol should prevent unregistered parties from bypassing identity
authorization.  In particular, any unauthorized party should not be allowed to
trigger channel closure and the potential dispute processes.  Moreover, any
unauthorized payment state should not be verified by a channel contract
successfully.  In addition, we aim the watchtower service to be fail-safe, i.e.,
an unavailable watchtower (e.g., attacked via a (D)DoS) should not silently
accept channel closures and its default behavior should be rejective (in
practice, a long timeout should be triggered in that case as otherwise parties
would not recover from an unavailable watchtower).

\paragraph{\textbf{Efficiency and Low Cost.}}
Our protocol should operate as efficiently as possible. 
There should be no 
efficiency bottlenecks with respect to throughput, storage, latency, etc. 
which could hinder its applicability in real-world scenarios.
Moreover, the cost of applying proposed solutions
in terms of storage, computation, and blockchain-related fees should be minimized. 

\paragraph{\textbf{Scalability.}}
Watchtowers should be able to handle a large number of payment channels and the
overheads introduced should not increase significantly with
the increasing number of channels and transactions. 

\paragraph{\textbf{Privacy.}}
Transaction details should be known only to transacting parties. 
In particular,
watchtowers should not learn transactions or their patterns while protecting the
channel. 
We emphasize, that in the case of dispute or channel closure, privacy of these
``closing'' transactions
may not be achieved as these actions are conducted over a publicly readable
blockchain.

\paragraph{\textbf{Accountability}}
Watchtowers should fulfill duties once received payment from parties. 
In particular, employed watchtowers should collect messages and record the latest state for 
a given channel. 
They also need to deal with closure events, resolve potential dispute,
and proceed payout in time. 
Otherwise, an honest party should always be able to prove 
the wrongdoing of a watchtower and withdraw payment.

In the following sections, we will introduce the details and interactions of our protocol.
\subsection{Channel Setup}
\label{sec:setup}
We assume two mutually untrusting parties, Alice and Bob, that wish to establish
and use a bidirectional payment channel. 
To this end, one of them has to deploy
a payment channel contract on the blockchain as shown in \autoref{alg:paymentcontract}.
Afterwards, the unique address \texttt{cid} of this payment channel contract is
used to reference the payment channel. 
The life cycle of the channel contract
can be divided into three phases ($\bot$, \texttt{OK}, and \texttt{DISPUTE}),
which is indicated by a global flag of the contract.

Alice and Bob then register their public keys ($pk_A$, $pk_B$) 
and deposit coins
into the contract for their subsequent transactions within the channel.
To achieve this, one party (say Alice) is responsible for initializing channel
using the \texttt{setup()} method.  
Before the actual setup, she is supposed to receive a signature from 
the counter-party (say Bob) where the signed message is the initial state ($s_0$) between 
parties and the $bal_A$ in $s_0$ represents the amount Alice is depositing. 
Then, she deposits $bal_A$ and assigns a watchtower by 
providing the address of a
watchtower contract. 
This watchtower contract is maintained by a third party
who will be employed to monitor the channel and prevent it
from finalizing with outdated states. 
The flag transitions from $\bot$ to \texttt{OK}.
Afterwards, Bob also needs to get the signature from Alice that signs
a new exchange state ($s_1$) and then he freezes $bal_B$ in $s_1$ via the \texttt{deposit()} 
method
shown in~\autoref{alg:paymentcontract}.
These atomic exchanges enjoy two advantages: 1) ensures the \texttt{amount}
deposited by each party is signed and agreed by the counter-party, 2) avoids
clash behavior from any party (in the case of failed setup any party can get its
deposit back). 
For an honest party, we allow her to get coins back
via directly invoking close request and inputting signed state.

\subsection{Watchtower Employment}
\label{sec:linear}
Before Alice and Bob can securely make off-chain transactions via the
payment channel, they need to employ a watchtower to monitor the channel. 
One of the participants pays the watching service via a watchtower 
contract shown in~\autoref{alg:wtcontract}.   

Although, watchtowers can be introduced to a channel in many ways, for a simple
description we assume
that Bob first deposits
\texttt{amount} in the tower contract and then sends a signature with \texttt{cid} 
to the watchtower which
verifies the authenticity of the payment by simply checking if Bob
deposits the correct amount. 
As a result, the payment channel with \texttt{cid} is protected by the
watchtower. 
A watchtower is considered trusted, however, a customer (i.e., a party) should be able
to withdraw \texttt{amount} if he proves watchtower is wrongdoing or not fulfilling eligible
job. 
To successfully withdraw a deposit, a customer should submit a \texttt{challenge} 
via the channel contract that proves the watchtower did not either close the given channel with 
the latest state or respond to the channel contract in time. 
We employ a tolerance timeout $t$ motivating the watchtower to respond faster and
more efficiently. 
The watchtower will not lose anything if it replies to the channel contract
before $t$. 
Otherwise, a longer timeout $T$ is initialized and a \textit{linear rewarding} would be triggered.
This mechanism rewards the watchtower inversely proportional to the watchtower's
response time.
It prevents unavailable or failed watchtowers from 
getting undeserved fees and incentivizes watchtowers to maintain high
availability and to shorten timeouts.

\subsection{Regular Payment}
\label{sec:microexchange}
After a watchtower is employed, Alice and Bob are ready to make
off-chain transactions via the payment channel by participating
in an interactive protocol depicted in \autoref{fig:initial}.

\begin{figure}[tb!]
\centering
\includegraphics[width=0.7\linewidth]{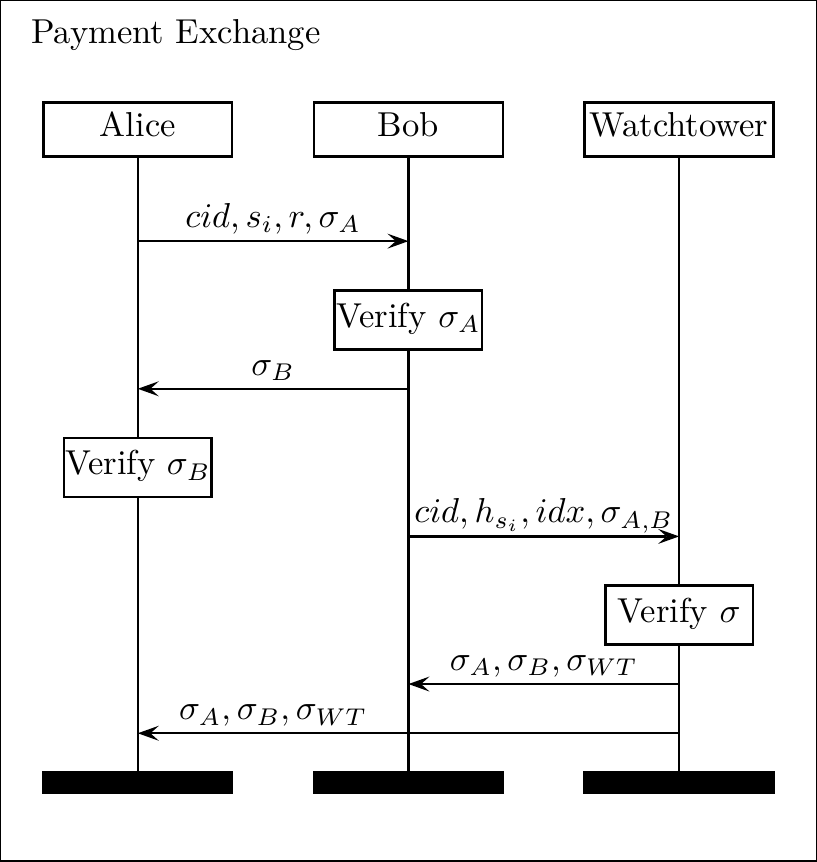}
\caption{Payment Exchange}
\label{fig:initial}
\end{figure} 

As in~\cite{dziembowski2017perun}, the payment channel is modeled as a state machine,
and the state transitions from $s_i$ to $s_{i+1}$ 
are carried out by off-chain transactions signed 
by both Alice and Bob. 
We encode each state $s_i$ as a tuple
$(bal_A, bal_B, idx)$, where
$bal_A$ and $bal_B$ are the corresponding balances of Alice and Bob within the channel,
and $idx$ is a monotonic index starting from $0$, which is incremented for every state transition
and the state with the largest $idx$ is considered the latest state.

To make a new transaction, Alice computes
the hash value 
$h_{s_i} = \texttt{H}(s_i \parallel r)$
with recalculated balances ($bal_A$, $bal_B$) and incremented $idx$, where 
$r$ is a large (e.g., 256-bit) random number. 
Then she sends the signature 
$$\sigma_A = \texttt{Sign}(sk_A,  \texttt{cid} \parallel  idx \parallel h_{s_i})$$
together with the message $\texttt{cid} \parallel s_i \parallel r$
to Bob. 
Next, Bob verifies the signature and authorizes the transaction
by sending his own signature
$\sigma_B = \texttt{Sign}(sk_B, \texttt{cid} \parallel  idx \parallel h_{s_i})$
to Alice. 
Meanwhile, he verifies the signature $\sigma_A$ and sends
$\texttt{cid}$, $idx$, $h_{s_i}$, $\sigma_A$ and $\sigma_B$ to
the watchtower. 
After successful verification of the 
signatures, the watchtower records it as the latest state for
the given payment channel and delivers a receipt
$\sigma_{{WT}} = \texttt{Sign}(sk_{WT}, \texttt{cid} \parallel  idx \parallel h_{s_i})$
to both Alice and Bob. 
The watchtower will not sign and send newer state to each party once
any closure or dispute happened for avoiding race conditions.

If all parties follow the protocol honestly, 
then all computations and message exchanges can be 
performed off-chain. 
Moreover, the watchtower does
not learn about the balances of Alice and Bob since it only receives
the hash value $h_{s_i}$. 

\begin{algorithm}[!t]
\setlength{\abovecaptionskip}{1.0cm}
    \caption{Payment Channel Contract.}
    \label{alg:paymentcontract}
    \SetAlgoVlined
    $Flag: \bot, \texttt{OK}, \texttt{DISPUTE}, s: NULL$\\
    $ddl: 0, isRspd: NULL, perc: 0, end: 0$\\
    $event~Closure(cid, state, r), event~Dispute(cid, state, r)$\\
    \medskip
    \SetKwProg{func}{function}{}{}
    \func{setup($party, \$bal, wtAddr, wt$)}{
        $assert(Flag = \bot)$\\
        $pk_A \leftarrow party,~~bal_A \leftarrow \$bal, pk_{wt} \leftarrow wt$\\
        $watchTower \leftarrow TContract(wtAddr),~Flag \leftarrow OK$\\
    }
    \func{deposit($party, \$bal$)}{
        $assert(Flag = OK),~pk_B \leftarrow party,~~bal_B \leftarrow \$bal$\\
    }

    \func{close($s_i, r, \mathcal{P}_{\sigma_k} $)}{
        $assert(Flag = OK~and~verify(pk_k, (s_i,r), \mathcal{P}_{\sigma_k}))$\\
        $Flag \leftarrow DISPUTE, s \leftarrow s_i$\\
        $ddl \leftarrow now + t, end \leftarrow ddl + T, isRspd \leftarrow False$\\
        $watchTower.close(this, s_i),emit~Closure(this, s_i, r)$\\
    }

    \func{dispute($s_i, r, \mathcal{P}_{\sigma_k}$)}{
        $assert(verify(pk_k, (s_i,r), \mathcal{P}_{\sigma_k})~and~s_i.idx > s.idx)$\\
        $assert(Flag = DISPUTE~and~now < end), s \leftarrow s_i$\\
        $watchTower.close(this, s_i), emit~Dispute(this, s_i, r)$\\
    }

    \func{payout($s_i, isPay$)}{
        $assert(Flag = DISPUTE)$\\
        \uIf{$msg.sender = watchTower$}{
            \If{$now > ddl~and~!isRspd$}{$perc \leftarrow (now - ddl)/ T$}
            $isRspd \leftarrow True$\\
            \eIf{$isPay = True$}
            { 
            $assert(bal_A + bal_B \geq s_i.bal_A + s_i.bal_B~and~s_i = s)$\\
            $partyA.transfer(s_i.bal_A), partyB.transfer(s_i.bal_B)$\\
            $Flag \leftarrow \bot, s \leftarrow s_i$}
            {$end \leftarrow now + T$}
        }
        \uElseIf{$msg.sender = any~party$}
        { 
            $assert(now > end~and~s_i = s)$\\
            $assert(bal_A + bal_B \geq s_i.bal_A + s_i.bal_B~and~s_i = s)$\\
            $partyA.transfer(s_i.bal_A), partyB.transfer(s_i.bal_B)$\\
            $Flag \leftarrow \bot, s \leftarrow s_i$}
    }

    \func{challenge($s_i, r, \sigma_{wt}$)}{
        $assert(verify(pk_{wt}, (s_i,r), \sigma_{wt})~and~now > end)$\\
        \If{($Flag = \bot~and~s_i.idx > s.idx$)~or~($isRspd=False$)}
        {$watchTower.withdraw(this, msg.sender, 1)$\\}
        \If{$perc > 0$}{$watchTower.withdraw(this, msg.sender, perc)$\\}
    }

\end{algorithm}

\begin{algorithm}[t!]
    \caption{Tower Contract.}
    \label{alg:wtcontract}
    \SetAlgoVlined
    $mapping (integer \Rightarrow \textlangle cidArray[], stateArray[]\textrangle )~Channels$\\
    $mapping (address \Rightarrow \textlangle customer, deposit\textrangle)~Balances, k \leftarrow 0$\\
    \medskip
    \SetKwProg{func}{function}{}{}
    \func{deposit($cid$)}{
        Balances[cid].customer $\leftarrow$ msg.sender;\\
        Balances[cid].deposit $\mathrel{{+}{=}}$ msg.value;
    }
    
    \func{withdraw($cid, victim, percentage$)}{
        assert(victim = Balances[cid].customer)\\
        assert(cid = msg.sender)\\
        \If{$percentage > 1$}{$percentage \leftarrow 1$}
        victim.transfer(Balances[cid].deposit $\times$ percentage)\\
    }


    \func{close($cid, s_i$)}{
        \texttt{/*find() returns the index of $cid$ }\\
        \texttt{\quad in cidArray, -1 means non-exist*/}\\
        $cidIndex \leftarrow Channels[k].cidArray.find(cid)$\\
        \eIf{$cidIndex = -1$}
        {$Channels[k].cidArray.push(cid)$\\$Channels[k].cidArray.push(s_i)$}
        {   
            $s \leftarrow Channels[k].stateArray[cidIndex]$\\
            $Channels[k].stateArray[cidIndex] \leftarrow s_i$}
    }
    
    \func{update($Confs$)}{
        $assert(msg.sender = owner)$\\
        $k \leftarrow k + 1$\\
        $respond(Confs, k-1)$\\
    }

    \func{respond($Confs,n$)}{
        \For{$0 \leq  j < Channels[n].cidArray.length()$}{
            $s \leftarrow Channels[$n$].stateArray[j]$\\
            \eIf{Confs[j] = 1}    
            {Channels[$n$].cidArray[j].call(payout($s$, True))} 
            {Channels[$n$].cidArray[j].call(payout($s$, False))}
        }
        $delete~Channels[n]$\\
    }
\end{algorithm}

\subsection{Watchtower Contract Management}
Tower contracts are deployed and managed by watchtowers.
A tower contract (see \autoref{alg:wtcontract}) maintains 
two mapping (or dictionary) data structures: 
\texttt{Balances}, \texttt{Channels}, and a bitmap structure
confirmation set (\texttt{Confs}). 

Items in \texttt{Balances} can be referenced by \texttt{Balances[cid]}, where
\texttt{cid} is the address of a channel contract. 
There are two fields
for each item and they can be accessed via \texttt{Balances[cid].customer} and
\texttt{Balances[cid].deposit}, which record the account (public key) of 
an external owned account and its deposit respectively.

The data structure \texttt{Channels} with integer keys stores tuples
$<cidArray[], stateArray[]>$
whose entries are two arrays, where $stateArray[i]$ is supposed to keep track
of the state of the channel contract with address (or cid) $cidArray[i]$.  

\texttt{Confs} stores a bit-vector (or a \textit{bitmap}) 
whose $j$th bit indicates whether the
channel with address
\texttt{Channels[k].cidArray[j]}
is allowed to be closed, where $k$ is an update counter.  
Such a decision is encoded via a single bit only.
The \texttt{Confs} data structure is periodically updated by a watchtower. 
A rational watchtower wishing to maximize its profit 
is expected to respond to
a channel contract before a tolerance timeout $t$ as otherwise its fee will
decrease linearly with the response delay.
Since a single watchtower can be employed
by multiple channels, during each update period, a
watchtower is expected to capture multiple channel-closing 
transactions through the \textit{closure events} emitted by
the channel contracts, containing the values of \texttt{cid},
\texttt{$s_i$}, and a random number $r$, which are necessary for
the watchtower to verify the validity of the intended closing state. 

Let us assume that during the update period, the watchtower
captures a sequence of closure events for channels with addresses
$$[cid_0,cid_{1}, \cdots,cid_{m-1}].$$
For each channel, the watchtower extracts relevant
information from the corresponding event and checks
whether the closing state \texttt{$s_i$} is compatible with its
local record of the latest state of the channel.
As a result, a confirmation set 
$[b_0, b_1, \cdots, b_{m-1}] \in \{0, 1\}^m$
is constructed such that $b_j = 1$ if and only
if $cid_j$ is allowed to be closed. 
Then the watchtower updates \texttt{Confs}  
in the tower contract by calling its \texttt{update()} method
with the newly constructed confirmation set. This triggers 
the \texttt{respond()} method which in turn invokes 
the \texttt{payout()} method of the channel contract.

\subsection{Payment Closure}
\label{sec:terminate}
Any payment channel participant who wants to finalize the channel 
with a specific state $s_i$ can initiate the channel closure process
by invoking the \texttt{close()} method of the channel contract presented 
in~\autoref{alg:paymentcontract}.
After successfully verifying the supplied signatures together with the 
associated state $s_i$ and the random number $r$, the channel contract changes its flag
to \texttt{DISPUTE} and sets a tolerance timeout $t$ from now on. 
Afterwards, the \texttt{close()} method of the watchtower contract can be called
with the state $s_i$.

Then the watchtower contract updates the mapping data structure (\texttt{Channels})
by placing the cid and the state $s_i$ into $cidArray$ and $stateArray$ or
only updating $stateArray$ when cid already exists in the $cidArray$
(see line 15 - line 21 of \autoref{alg:wtcontract}).

The interaction between the channel contract and the tower contract
is finished when the watchtower calls 
its \texttt{update()} method to update the confirmation set structure.
Then, the \texttt{update()} method invokes the \texttt{respond()} method of the 
tower contract to
determine which channels are allowed to close according to values in \texttt{Confs}.
The \texttt{respond()} method will call the \texttt{payout()} 
method of all channels appearing in the 
\texttt{Channels} data structure
with the decisions on if they are allowed to be closed. 
\medskip

To illustrate different cases of this process, we show an example
scenario, where we assume that Alice and Bob have exchanged three payment 
tuples:
$$T_1:(10, 0, 0), T_2:(7, 3, 1), T_3:(4, 6, 2).$$
As depicted in \autoref{fig:terminate},
Alice closes $cid_i$ with \texttt{$T_3$}. 
Once noticed, 
a tolerance timeout $t$ (e.g., one hour) would be triggered in the channel contract
that allows the tower contract to confirm the submitted state and to proceed 
\texttt{payout} and consequently close the channel.
The tower contract replies to the channel contract based on a confirmation set
updated by the watchtower periodically. 
In particular, $cid_i$ would be allowed to close only if 
the corresponding index is set to \textit{1} in the confirmation set.

\begin{figure}[tb!]
    \centering
    \includegraphics[width=1.0\linewidth]{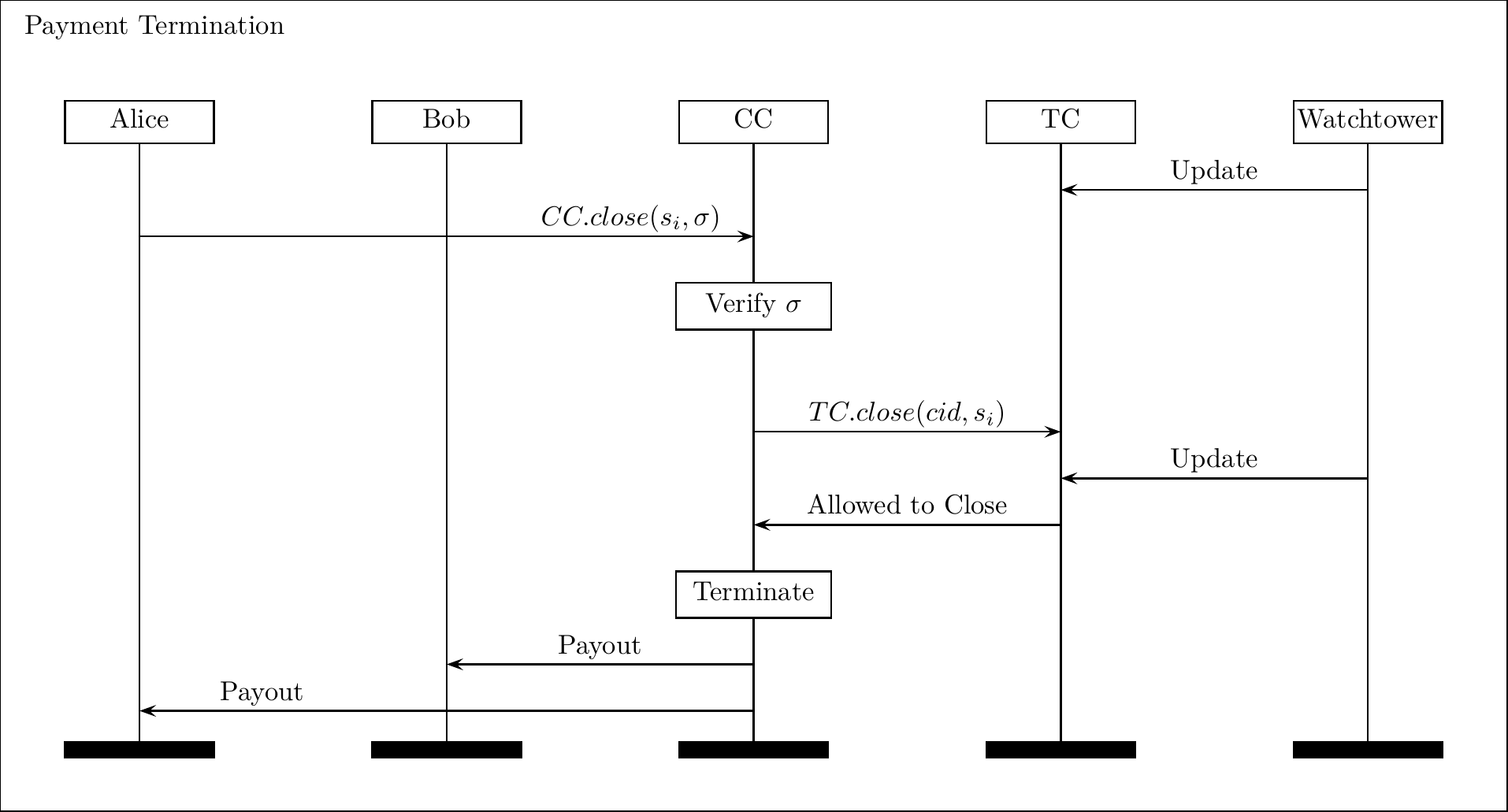}
    \caption{Payment Termination.}
    \label{fig:terminate}
\end{figure}

An alternative case is when Alice wants to misbehave by submitting the stale
state of
\texttt{$T_2$} and potentially benefiting due to her higher balance.
Therefore, she invokes a new close event passing \texttt{$T_2$} as the current
state.
However, in such a case, an honest watchtower
would deny the closure request since \texttt{$T_2$} is
a stale state. 
Therefore, the watchtower  sets \textit{0} in the corresponding position of 
the confirmation set for $cid_i$.
Once disallowed, 
a longer and fail-safe timeout \textit{T} (\textit{T} $\gg$ 
\textit{$t$}, e.g., a day or two) is initialized that ensures that the 
deposit will not be paid out if any misbehavior happened. 
A similar situation will happen
when the watchtower does not respond on time and \textit{T} is also initialized.
We assume that parties can connect to the blockchain at least once per
\textit{T}, thus Bob would notice the misbehavior performed by Alice.
Then Bob commits a 
dispute transaction with \texttt{$T_3$} as an input ,
aiming to prove that \texttt{$T_2$} is an older off-chain payment state. 
In such a case, the watchtower would resolve the potential dispute and close the channel faster,
namely, before \textit{T}.
Otherwise, \texttt{payout} will be proceeded by any party
after \textit{T} units of time. 

By this design, our construction minimizes waiting times and achieves fail-safe, 
as in common case
a watchtower would confirm payouts quickly.
Meanwhile, the long timeout \textit{T} prevents
any misbehavior from parties or unavailability and wrongdoings from the watchtower 
but without additional storage overhead introduced compared with existing 
watching service schemes.
Moreover, the watchtower can minimize costs confirming multiple pending channel closure requests
only with sending a single blockchain message.

\section{Implementation and Evaluation}
\label{sec:analysis}
To evaluate our design, we implemented the payment channel contract and the tower contract
using the Solidity programming language
(compiler version 0.4.24)
and we deployed our scheme on an Ethereum testnet. 
The off-chain watchtower and two parties, Alice and Bob, are 
implemented as a web server running
\texttt{Node.js v10.2.1} bundled 
with the \texttt{node-localStorage} package 
for storing \texttt{cid}, \texttt{$s_i$}, \texttt{Hs}, and signatures. 
We employ 
\texttt{web3.js}~\footnote{\url{https://web3js.readthedocs.io/en/1.0/}} to
interact with deployed smart contracts.
We use the Ethereum's ECDSA signature scheme as the default one, as Ethereum
provides a native and optimized support for it.

\subsection{Cost}
\label{sec:cost}
In our protocol, to invoke the closure procedure, 
a party sends a transaction with a proper state \texttt{$s_i$} 
which is verified first by the channel contract, 
then the watchtower interacts with the tower contract which handles the closure event. 
The main cost is introduced due to
the computation and storage whose utilization 
is charged by the Ethereum network. 
We perform a series of experiments to measure the cost 
in terms of gas consumption and its corresponding monetary cost.
\begin{table}[!tb]
  \caption{Channel closure cost (per request).}
  \label{table:closurecost}\centering
\begin{tabular}{L{1.0cm}R{1.5cm}R{1.2cm}R{1.2cm}R{1.2cm}}
\toprule
\multirow{2}*{Cost} & \multicolumn{4}{c}{Number of Closure Request}\\ \cmidrule{2-5}
          & \texttt{$10^0$}  &\texttt{$10^1$}  &\texttt{$10^2$} &\texttt{$10^3$} \\ \midrule
Verify    &57561 (14.0\%)    &57572 (19.1\%)   &57565 (20.6\%)   &57561 (20.7\%)  \\ 
Storage   &151849 (37.0\%)  &127812 (42.5\%)  &113802 (40.8\%)  &113793 (40.9\%) \\ 
Update   &84130 (20.5\%)    &8523 (2.8\%)    &860 (0.3\%)     &167 (0.1\%) \\ 
Payout   &35001 (8.5\%)   &34469 (11.5\%)    &34472 (12.3\%)   &34462 (12.4\%) \\ 
Misc      &82354 (20.0\%)  &72409 (24.1\%)    &72510 (26.0\%)   &72487 (25.9\%) \\
\bottomrule
Total   &410895         &300785           &279209         &278470   \\
USD       &0.198            &0.145             &0.135           &0.134         \\ \bottomrule
\end{tabular}
\end{table}

\paragraph{\textbf{Closure.}}
We conduct experiments for a different number of channel closure requests,
measure their gas cost and present the results, together with the cost converted to US dollars~\footnote{The conversion was according
to the gas price from \url{https://ethgasstation.info/} at the time of writing
the paper.}, in \autoref{table:closurecost}. 

We measure the cost of a channel closure for a different number of requests,
i.e., $10^{i}$ ($0 \leq i \leq 3$) requests. 
Conservatively, we assume that all the closure
requests are submitted
during the $t$-long time period (i.e., all these requests will be
processed using a single confirmation set).
From \autoref{table:closurecost}, we can see that the 
dominating operation cost is due to the storage required for \texttt{cid}
addresses and states \texttt{$s_i$}.
The cost also slightly decreases with the number of closure requests due to data
structure initialization cost that is paid upfront when initializing a tower contract.
Similarly, the cost of updating a confirmation set also shows a decreasing trend per request
since a watchtower updates it in batches periodically. 
The more closure events happen in
$t$ period the less each channel has to pay due to they share equally. 
Therefore,
as depicted the concept of \textit{fail-safe watchtowers} is cost-effective for handling a large number of channels. 
The payout procedure is the final step of the whole closure process which has a stable cost
with the channel number scaled.
The overall cost of each channel closure request is around \$0.14.

\paragraph{\textbf{Dispute.}}
\begin{table}[!h]
\centering
\caption{Cost for dispute.}
\label{table:disputecost}
\begin{tabular}{lrrrrr}
\toprule
\multirow{2}*{Cost} & \multicolumn{3}{c}{Steps} &\multirow{2}*{Total} \\ \cmidrule{2-4} 
& \texttt{$V_{state}$}  &\texttt{$V_{time}$}  &\texttt{$Closure$} 
\\ \midrule
Gas     &22752   &30337      &326123     &379212  \\
USD      &0.011    &0.015    &0.157  &0.183       \\ \bottomrule
\end{tabular}
\end{table}
As discussed in \autoref{sec:terminate}, 
our design allows any party to invoke a
potential dispute process. 
The gas cost of the dispute process is similar,
despite additional verification operations are required on the channel contract side.
We conduct the analogical experiments as in the previous case and the
results are shown in \autoref{table:disputecost}.
To successfully invoke a dispute, a party needs to submit a newer signed and valid
state \texttt{$s_i$}. 
The channel contract
also checks if current time is illegal before verifying the flag and validation of
\texttt{$s_i$}). 
From \autoref{table:disputecost}, we can see that the most common closure
case still
dominates the main cost within the dispute operation but the entire cost is
reasonably low at \$0.183.

\paragraph{\textbf{Confirmation Set.}}
As discussed in \autoref{sec:terminate}, a watchtower sends a single message to
the
blockchain network, creates an updated confirmation set periodically which requires to be stored
on the channel contracts side. 
The size (and consequently the cost) of this storage depends on the 
number of channel-closing requests. 
To give insights on the cost, 
we measure the cost of updating a confirmation set for the number of 
$10^{i}$ ($1 \leq i \leq 3$) channels
and present the obtained results in \autoref{table:updatecost}. 
As shown, 
the cost for the confirmation set
update operation with a small number of channels increases slightly. 
With the
number of channels growing, our protocol shows a good characteristics where the
only cost is around \$0.1 for $10^3$ channels per update transaction.
Such a low cost is achieved mainly by encoding every channel closure by only a
single bit.
\begin{table}[!h]
\centering
\caption{Cost for set update.}
\label{table:updatecost}
\begin{tabular}{lrrrr}
\toprule
\multirow{2}*{Cost} & \multicolumn{3}{c}{Number of Channel}\\ \cmidrule{2-4}
          & \texttt{$10^1$}  &\texttt{$10^2$} &\texttt{$10^3$} \\ \midrule
Total (gas)     &85232   &85990    &167363   \\
USD       &0.043    &0.044    &0.085         \\ \bottomrule
\end{tabular}
\end{table}

\subsection{Throughput}
\label{sec:throughput}
We evaluate the watchtower throughput running an instance on a system with macOS Sierra 10.12.6, Intel Core i5 CPU (2.7 GHz), and 8GB of RAM. 
We create $10^i$ ($0 \leq i \leq 4$) off-chain
transactions between parties,
each of them exchanges a state with a watchtower, then we measure the total time needed by the watchtower 
to handle the exchange steps, and compute the average time required per transaction. 
The obtained results are summarized in \autoref{fig:throughput}, where 
we can see that the number of
off-chain transactions handled per second raises 
when requests are processed in batches.
The throughput becomes stable when the number of requests is greater than $10^4$,
with the 
total time cost around 3.95ms per off-chain exchange.
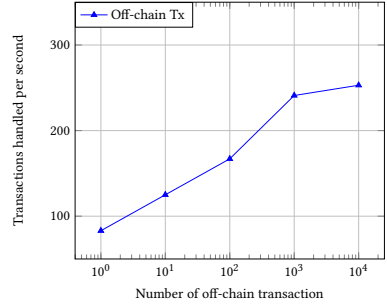
\begin{figure}[!h]
\centering
\begin{tikzpicture}[scale = 0.6,>=stealth]
\begin{axis}[
    legend cell align={left},
    ymin=50,
    ymax=350,
    grid=major,
    xmode=log,
    yticklabel style={
            /pgf/number format/fixed,
            /pgf/number format/precision=2,
            /pgf/number format/fixed
    },
    legend style={
    at={(0,0)},
    anchor=north west,at={(axis description cs:0,1.0)}},
    scaled y ticks=false,
    xlabel={Number of off-chain transaction},
    ylabel={Transactions handled per second}
]
\addplot[mark=triangle*, blue] coordinates {
    (1,83) (1e1,125) (1e2,167) (1e3, 241) (1e4,253) 
};
    \legend{Off-chain Tx}
\end{axis}
\end{tikzpicture}
  \caption{Throughput of the Watchtower.}
  \label{fig:throughput}
\end{figure}


\subsection{Efficiency Comparison}
In this section, we compare Lightning Channels, Duplex micropayment channels,
Raiden, PISA, and 
our work. 
All comparisons are based on the protocol descriptions presented in the context of \autoref{sec:related}. 
We focus on the number of signatures/transactions required for each step in the above protocols 
and the corresponding cost incurred.

\paragraph{\textbf{Signatures and Transactions.}}
\autoref{table:comparesig} highlights the number of signatures required for each
stage of Duplex~\cite{decker2015fast}, Lightning~\cite{lightening},
Raiden~\cite{Raiden}, PISA~\cite{mccorry2018pisa}, and our work. 
Next, we provide a comparison on the number of signatures and transactions that are required
for each stage of the protocols. 
For the fairness of the comparison, we assume that
only two participants are involved in all protocols.

Lightning and Duplex require both participants to sign the channel's state
first,
before cooperatively signing a so-called \textit{Funding Transaction}. 
The state in Lightning is simply a \textit{pair} of commitment transactions, 
thus there should be \textit{2$\times$2} transactions required in total.
While in Duplex the state is the first branch of an \textit{invalidation tree} which 
consists of $d$ nodes and two unidirectional channels. 
In Raiden and our work, the channel setup stage requires one party to sign a transaction to 
establish the channel with deposit. 
The counter-party can sign another transaction 
that aims to deposit coins and register identity by himself only. 
To initialize a channel in PISA, a customer invokes a
setup transaction with registering two parties each of whom 
will respectively deposit coins afterwards.
\begin{table}[!tb]
\centering
\caption{Comparison on the number of signatures/transactions required for each stage.}
\label{table:comparesig}
\begin{tabular}{lrrrrrr}
\toprule
\multirow{2}*{Cost} & \multicolumn{3}{c}{Stage}\\ \cmidrule{2-4}
          & \texttt{Setup}  &\texttt{Payment}    &\texttt{Dispute} \\ \midrule
Duplex~\cite{decker2015fast}     &$(d+2)\times2$      &1               &$1 \times 2$    \\
Lightning~\cite{lightening}      &$2\times 2$          &$1\times2$      &$\geq 3$       \\ 
Raiden~\cite{Raiden}             &2                   &1               &$\geq 2$       \\ 
PISA~\cite{mccorry2018pisa}      &$1+2$                 &3               &$\geq 3$       \\
Our work                         &2                   &3               &2       \\ \bottomrule
\end{tabular}
\end{table}

As for the payment stage,
Raiden and Duplex are used for a pair of unidirectional channels which
    require a single signature generated from a party to perform a newer payment.
The current balance of both parties should be maintained in Lightning. To update
their own balances and send an off-chain payment, 
both customers need to sign a new transaction.
Payment exchange in PISA requires one party to separately sign a
    \textit{command} message 
and new \textit{$state_i$}. 
Both signatures $\sigma^{cmd}_{k}$ and
$\sigma^{state}_{k}$ are sent to a counter-party who responds with a new agreed $\sigma$ then. 
In our work, three signatures are required for every payment, including two from
the transacting parties and another one created by the watchtower.

Our work shows a good performance in the case of
a dispute. 
Lightning requires the disputer to send commitment transaction first which
is followed by two transactions where each party claims its final balance and 
all unlocked conditional transfers. 
In some cases, the conditional transfers
can be associated with multiple lock times. 
Therefore, an additional transaction
to claim each conditional transfer would be incurred which requires $\geq$3
transactions in total.
In Duplex, both parties must sign the states of the unidirectional channels
that represent their latest payment from the counter-party and 
send both transactions to the blockchain. 
Similarly, each party in Raiden sends the same two transactions to the network as
in Duplex.  
However, each conditional transfer that needs to be unlocked requires 
the receiver to sign an additional transaction to claim it, which takes $\geq$2
transactions to settle the dispute.
When a dispute occurs in PISA, any party within the channel can enforce an on-chain
state transition via the dispute process. 
First, one party updates the state via \textit{setState} then initiates the dispute using 
\textit{triggerDispute}. 
Moreover, PISA allows all parties to input a command \textit{cmd} to be considered for the
state transition which requires $\geq$3 transactions in aggregate.
In our protocol, one party signs a single transaction to enforce a dispute. 
The channel will wait for a decision from the tower contract managed
by the watchtower who sends another single message to resolve the dispute.

\paragraph{\textbf{Cost.}}
\begin{table}[!t]
\centering
\caption{Comparison on the cost for each step.}
\label{table:comparecost}
\begin{tabular}{lrrr}
\toprule
\multirow{2}*{Step} & \multicolumn{3}{c}{Cost}\\ \cmidrule{2-4}
          & \texttt{Cost (PISA)}  &\texttt{Cost (Ours)}    &\texttt{USD} \\ \midrule
\multicolumn{4}{l}{Create Contract} \\ \midrule
Channel contract     &1892135      &1576314                    &0.761    \\
Custodian contract   &1609613      &\xmark                     &N.A.      \\ 
Tower contract       &\xmark       &1314215                    &0.635      \\ \midrule
\multicolumn{4}{l}{Interactions} \\ \midrule
Channel setup    &65504        &77328                      &0.037       \\ 
Withdraw employment  &89182               &86347               &0.042      \\\midrule
\multicolumn{4}{l}{Dispute Process} \\ \midrule
Party set state       &90130                   &\xmark               &N.A.       \\ 
Party initiates dispute  &78667                &\xmark               &N.A.       \\ 
Party submit command   &140275                 &\xmark               &N.A.       \\ 
Transit state on-chain &149494                 &\xmark               &N.A.       \\ 
Total resolve dispute   &458566                &379212              &0.183       \\ \bottomrule
\end{tabular}
\end{table}
To further compare our scheme with the most related work, we conducted
experiments using the PISA implementation.
\autoref{table:comparecost} presents the detailed cost comparison of our protocol and PISA~\cite{mccorry2018pisa}.
The result shows that one-time deployment cost of our channel contract is lower than
that of PISA since we remove the state transition and input methods, and instead
employ a tower contract to handle closure and resolve dispute procedures. 
A watchtower periodically sends single message into tower contract where 
only lightweight data structures are deployed to store $cid$ and $s_i$. 
It is also cost-effective compared with sophisticated logic of the PISA's custodian contract.
In the channel setup phase, our protocol needs to store one more 
information (the address of the watchtower)
while only public keys of both parties are saved in PISA. 
The cost of the channel setup is slightly 
higher than in PISA but the difference is not significant.
Both, our work and PISA allow withdrawing deposited amount if parties are able
to prove a watching service's wrongdoing. A customer in PISA first seeks recourse
then invokes
the refund process but the cost of these operations is similar in both protocols.
When a dispute happens in PISA, one party first sets collectively authorized state
(\textit{setstate}) then initiates dispute process via \textit{triggerdispute}. 
Afterwards,
the party submits command list (\textit{input}) and eventually a state transition is
executed on-chain under the aid of custodian contract. 
Our work requires one party to enforce dispute that proves the newer 
payment state has been performed off-chain then employ tower contract to resolve.
The total cost is still lower than PISA.
In conclusion, our work can be regarded as more cost-effective.

\paragraph{\textbf{Payout Timeouts.}}
An important advantage of our design is that in the common case (i.e.,
parties are honest and a watchtower is available) the channel can be closed quickly.
To confirm it empirically, we conducted an experiment to measure a time delay
    required to
    close a channel. In our experiment the channel closures and the corresponding
payouts became part of the blockchain just after one or two blocks (i.e.,
between 14-28 seconds on average),  which is significantly faster than in systems
which silently accept closures after a timeout.  For instance, the timeouts in
Raiden refer to the number of blocks that are required to be mined from the
time that \textit{close()} is called until the channel can be settled with a
call to \textit{settle()} and its value is set by default to 50.  With block
arrival times of Ethereum  (14 seconds on average~\cite{blocktime}) the
required time in Raiden is around 12 minutes.  We emphasize, that usually it
is a parameter that parties can adjust; however, in the previous designs it
also constitutes the tolerable unavailability of watching services (thus, too
short timeout may cause the service ineffective).

\subsection{Deployment Considerations}
\label{sec:deployment}
\paragraph{\textbf{Scalability.}}
For simplicity and 
an intuitive description, we only give an example in the context of bidirectional payment 
channels but our protocol can be easily extended to multiple payment channels
and many users. 
To support channels with more than two users, the off-chain exchange (see
\autoref{sec:microexchange}) would be scaled without much more cost incurred
as the signature verification and communication expense are relatively cheap.
Moreover, watchtowers can easily scale to handle a large number of channels
and the confirmation set is easy to update and scale as well without increasing 
overhead drastically.
In practice, watchtowers can be realized without incurring significant overheads
and costs since they can be 
implemented as light blockchain clients. 
To have a concrete perspective on the cost, we measure
the communication and storage overhead for watchtower deployment. 
As depicted in~\autoref{fig:fromparties}, the data layout exchanged from parties to watchtower ($P\rightarrow WT$)
for each off-chain exchange is an $198$-bytes object, where \texttt{cid} indicates the observed 
channel address, $H_s$ encodes the hashed state value, $idx$ specifies the transaction index, 
$\sigma_A$ and $\sigma_B$ are the signatures from both parties.
Meanwhile, the results for the communication cost from watchtower 
to parties ($WT\rightarrow P$) and between parties ($P\Leftrightarrow P$)
are shown in~\autoref{table:development}.
Moreover, we measure the storage cost of storing the latest state value
for $10^{i}$ ($1 \leq i \leq 3$) channels. 
From~\autoref{table:development}, we can see that it takes only 12MB for a watchtower to observe
$10^3$ channels.

\begin{figure}[!ht]
\setlength{\belowcaptionskip}{-0.3cm}
    \begin{center}
        \begin{tikzpicture}[scale = 0.9,>=stealth]
        \draw (0, 0) rectangle node {\scriptsize $\texttt{cid}$} +(0.9, 0.5);  
        \node[scale = 0.7] at (0.45, 0.5+0.25) {20B};
        
        \draw (0.9, 0) rectangle node {\scriptsize $\texttt{$H_s$}$} +(1.2, 0.5);
        \node[scale = 0.7] at (1.5, 0.5+0.25) {32B};
        
        \draw (2.1, 0) rectangle node {\scriptsize $\texttt{idx}$} +(0.6, 0.5);
        \node[scale = 0.7] at (2.4, 0.5+0.25) {16B};
        
        \draw (2.7, 0) rectangle node {\scriptsize $\texttt{$\sigma_A$}$} +(2.4, 0.5);
        \node[scale = 0.7] at (3.9, 0.5+0.25) {65B};

        \draw (5.1, 0) rectangle node {\scriptsize $\texttt{$\sigma_B$}$} +(2.4, 0.5);
        \node[scale = 0.7] at (6.3, 0.5+0.25) {65B};
        
        \draw[dotted] (0,0) -- (0, -0.5);
        \draw[dotted] (7.5,0) -- (7.5, -0.5);
        
        \node[scale = 0.7] at (3.75, -0.25) {198 bytes};
        \draw[->] (3.75-1, -0.25) -- (0, -0.25);
        \draw[->] (3.75+1, -0.25) -- (7.5, -0.25);
        
        \end{tikzpicture}
        \caption{The layout of data from party to watchtower.}
        \label{fig:fromparties}
    \end{center}
\end{figure}
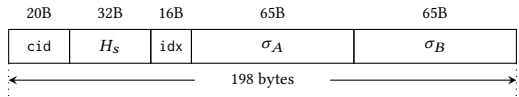

\paragraph{\textbf{Incentives.}}
The \textit{linear rewarding} mechanism (see ~\autoref{sec:linear})
motivates watchtowers to maintain high availability and to shorten payout
timeouts.
A mechanism punishing unavailable or failed watchtowers is also implemented
in our system where customers submit the ineligibility proofs to
the channel contract (see ~\autoref{alg:paymentcontract}) and eventually withdraw their deposit.
There are various options for parties to employ a watchtower according to their performance and commission,
which in turn incentivizes watchtowers to participant our system and provide better services.

\begin{table}[!tb]
  \caption{Watchtower Deployment Cost.}
  \label{table:development}\centering
\begin{tabular}{lrrr}
\toprule
\multirow{2}*{Cost} & \multicolumn{3}{c}{Number of channels}\\ \cmidrule{2-4}
                 & \texttt{$10^1$}  &\texttt{$10^2$}  &\texttt{$10^3$}\\ \midrule
Storage          &0.0011MB           &0.1103MB   &12MB      \\ \bottomrule
\multirow{2}*{Cost} & \multicolumn{3}{c}{Type of communications}\\ \cmidrule{2-4}
& \texttt{$P\rightarrow WT$}  &\texttt{$WT\rightarrow P$}  &\texttt{$P\Leftrightarrow P$}\\ \midrule
Communication   &198B          &195B   &165B \\ 
\bottomrule
\end{tabular}
\end{table}

\section{Security Discussion}
\label{sec:discussion}
%
First, we show that all successfully verifiable messages created within our protocol are authentic.
As we assume that any party cannot obtain a valid secret key
of its counter-party (see \autoref{sec:overview}), a malicious party cannot
generate and sign a payment state on behalf of honest parties. This implies that 
the newer states are agreed by all parties in our protocol. 
All channel closure, and optionally dispute, transactions require states verification
process first (see in \autoref{sec:terminate}) which
ensures that a state submitted by parties is authentic (i.e., signed by all participants). 
The only way to obtain such a state
is to be signed and exchanged by both parties within payment channels.

\medskip
Next, we formalize our system by the following definitions and we show its main
security properties.
\begin{definition}\label{definition:def1}
    By definition, the protocol~$\Pi(\mathcal{P}, \mathcal{W}, \mathcal{A}, \phi, \varphi)$ 
    is a dynamic process with five components involved, where $\mathcal{P}$ denotes
    the compliance parties in the payment channel, 
    $\mathcal{W}$ represents the watchtowers entities, 
    $\mathcal{A}$ is a potential adversary, and 
    $\phi$ and $\varphi$ specify the payment channel contract and the tower
    contract, respectively.
\end{definition} 

\begin{definition}\label{definition:def2}
    Let $\mathcal{C} = \{online,$offline$\}$ be a set of availability states for
    the compliance parties and watchtowers (please note that we consider a
    successful (D)DoS attack as the offline state). Then the function 
    $$
    0/1 \leftarrow \mathcal{F}_{secure}(\mathcal{P}_{\mathcal{C}}, \mathcal{W}_{\mathcal{C}}, \mathcal{A}, state_{i}, \sigma_{state})
    $$ 
    is triggered before the protocol~$\Pi(\mathcal{P}, \mathcal{W}, \mathcal{A}, \phi, \varphi)$ is finalized with $state_{i}$. $\mathcal{P}_{\mathcal{C}}$ and 
    $\mathcal{W}_{\mathcal{C}}$ reveal the availability state of the
    honest parties and watchtowers, 
    $state_{i}$ denotes an off-chain payment transaction submitted by any party
    while $\sigma_{state}$ is the corresponding signature. The output $1$ means that the
    $state_{i}$ is the latest off-chain payment state and the final payout
    should be performed,  
    while $0$ indicates that there are suspicious misbehavior (e.g.,
    an illegitimate dispute) and additional
    actions should be performed to prevent the system from going into an
    insecure state.
\end{definition}

\begin{theorem}\label{theorem:the1}
    The channel is allowed to close with a correct state submitted by honest $\mathcal{P}$.
\end{theorem} 

\begin{proof}
    When the parties registered in the channel agree to close the channel with $state_i$,
    $\mathcal{P}$ sends a closure transaction to $\phi$ and then the function 
    $\mathcal{F}_{secure}(\mathcal{P}_{online}, \mathcal{W}_{online},
    \mathcal{A}, state_{i}, \sigma_{state_{i}})$ is executed and returns $1$. 
    Afterwards, $\mathcal{W}$ checks the output of $\mathcal{F}_{secure}$, 
    updates \texttt{Confs} with $\varphi$ and confirms the final payout with $\phi$.
\end{proof}

\begin{theorem}\label{theorem:the2}
    If $\mathcal{A}$ plays a strategy that finalizes the payment channel with a stale state submitted
    when the counter-party $\mathcal{P}$ is in the offline state (e.g.,
    temporarily away or (D)DoSed),
    then $\mathcal{A}$ cannot gain any unearned income from the channel. 
\end{theorem}

\begin{proof}
    We assume that $\mathcal{A}$ invokes a \textit{closure transaction} by
    submitting the stale
    state $state_i$. Before finalizing the channel, the function $\mathcal{F}_{secure}(\mathcal{P}_{offline}, \mathcal{W}_{online}, \mathcal{A}, state_{i}, \sigma_{state_{i}})$
    outputs $0$ since an eligible watchtower rejects the state which is not the latest one.
    Then the long timeout \textit{T} is initialized. 
    In \autoref{sec:terminate}, we assume that parties can connect to the
    blockchain at least once per \textit{T}, thus when $\mathcal{P}$ become
    available it would notice the misbehavior performed by $\mathcal{A}$.
    Then $\mathcal{P}$ sends the dispute transaction into $\phi$ with the latest
    signed off-chain payment state. Eventually, $\mathcal{W}$ resolves the
    dispute and the payment channel finalizes with a correct state.
\end{proof}

\begin{theorem}\label{theorem:the3}
    If $\mathcal{W}$ is unavailable (i.e., offline for a while or under a (D)DoS
    attack  launched by $\mathcal{A}$), the protocol~$\Pi(\mathcal{P},
    \mathcal{W}, \mathcal{A}, \phi, \varphi)$ goes into a fail-safe state and
    the channel is still able to be finalized
    with a correct state.
\end{theorem} 

\begin{proof}
    The \autoref{theorem:the3} has two implications. 
    First, when $\mathcal{W}$ is offline for a while and becomes
    inaccessible,
    the function $$0 \leftarrow \mathcal{F}_{secure}(\mathcal{P}_{online}, \mathcal{W}_{offline}, \mathcal{A}, state_{i}, \sigma_{state_{i}})$$ as it cannot interact with $\varphi$ on time.
    Then a long timeout $T$ is triggered.
    In this period, the channel participant $\mathcal{P}$ 
    submits a dispute transactions into the blockchain. 
    Afterwards, $\mathcal{W}$ is responsible for
    resolving the dispute and closing the channel. 
    Second, a similar situation happens even when $\mathcal{W}$ undergoes a (D)DoS attack and becomes unavailable, in the period 
    of $T$, $\mathcal{P}$ performs dispute and $\phi$ records the latest 
    off-chain state. After $T$, the final payout would be proceeded
    by any party with the correct state.
\end{proof}

\paragraph{\textbf{MITM Attack or Substitution Attack.}}
To launch a MITM attack, $\mathcal{A}$ can eavesdrop and intercept transactions 
from $\mathcal{W}$ to $\varphi$ and $\mathcal{P}$ to $\phi$.
For instance, he extracts the \textit{Confs} and $s_i$ respectively, then constructs
new transactions with the intercepted data.
However, as shown at the beginning of this section, the malicious transactions cannot pass the verification procedures on $\phi$ and
$\varphi$'s side since $\mathcal{A}$ cannot compromise 
the underlying cryptographic primitives or
pretend to be the other parties.
Any tiny change of the context (i.e., $s_i$) will be 
caught by the signature verification process. 
Moreover, other eavesdropping or MITM attacks can be easily mitigated by deploying
an end-to-end authenticated and encrypted communication channel 
between the protocol parties. 

\paragraph{\textbf{Failure safety.}}
The concept of \textit{fail-safe watchtowers} ensures that honest 
parties are able to close channels with a quick timeout in most generic
cases. 
However, when an adversary is intentionally closing a channel with a
stale state, a longer but safe-assurance timeout is
introduced.
This long timeout avoids wrongdoing or unavailability of watchtowers. 
In particular, any unavailable or compromised watchtower should not accept channel closures
in silence. 
In contrast to competing systems, this design introduces longer timeouts only
for the worst-case (i.e., a misbehaving peer), while in the common case (i.e.,
honest peers and an available watchtower) channels are closed quickly and safe.

\paragraph{\textbf{State privacy.}}
Our protocol provides the so-called {\it state privacy} as proposed in~\cite{mccorry2018pisa}, 
which is an extension of the notion of {\it value privacy}~\cite{malavolta2017concurrency}.
Unless the details of the state are revealed by the channel participants in order to 
finalize the channel, a watchtower only learns the hash value $\texttt{H}(s_i
\parallel r)$.
Therefore, the metadata of all intermediate states are invisible to the watchtower,
thus the privacy of the off-chain transactions is protected. 
Note
that the hash value is computed with a large random number $r$, which prevents
a watchtower from finding the pre-image by exhausting $s_i$, whose entropy can be
fairly low.

\paragraph{\textbf{Availability.}}
Requiring watchtowers to keep off-chain verified states and updating
confirmation sets  
introduces a single point of failure, as with a failed (or (D)DoSed) watchtower it would not be
possible to respond to a channel closure request.
Fortunately, as discussed, in this case a longer timeout will be triggered.
Moreover, a watchtower instance can be easily replicated by replicas running a
traditional distributed consensus protocol to share current channel states.

\section{Short-lived Assertions}
\label{sec:assertions}
\subsection{Motivation}
\label{sec:motivation}
\begin{figure*}[t!]
    \centering
    \includegraphics[width=\linewidth]{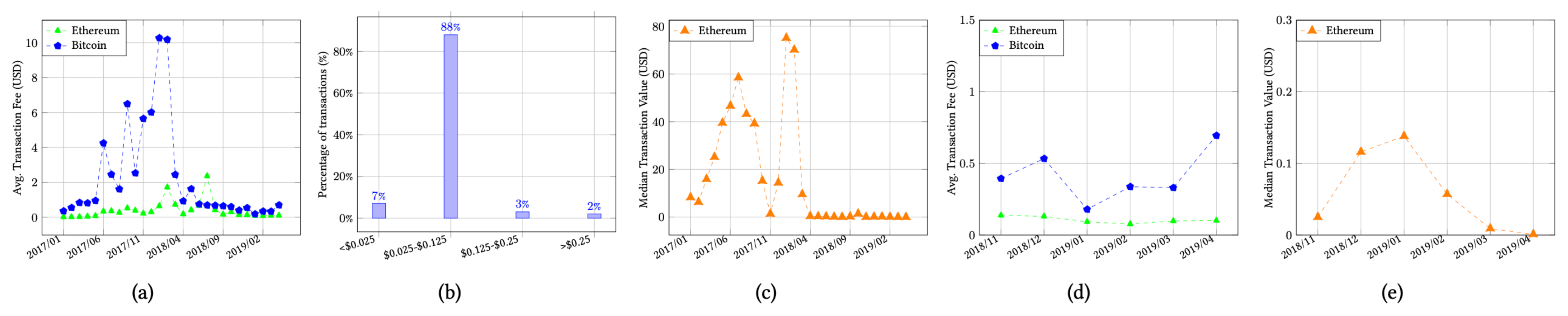}
    \caption{(a) Historical transaction fees of Ethereum and Bitcoin, (b) Bitcoin transaction fees 
versus frequency for 1 million transactions, (c) Historical transaction value of Ethereum,
(d) Average transaction fees of Ethereum and Bitcoin in recent 6 months, (e) 
Median transaction value of Ethereum in recent 6 months.}
    \label{fig:motivation}
\end{figure*}
Although watchtowers increase the security of payment channels, in small-amount
(micro)payment scenarios their deployment may be economically inefficient.
Traditional bank-based transactions usually incur between 21 to 25 cents (in the US) 
fees~\cite{pass_probchannel} plus a percentage of the transaction amount.
To better understand fee structures in blockchain platform, we investigated them
for Ethereum and Bitcoin.
For example, 
a plain transaction in Ethereum (without executing smart contracts) currently costs around
\$0.011~\footnote{The estimation is from \url{https://ethgasstation.info/}} and
similarly
in Bitcoin transaction fees are usually at 
least 0.0001 BTC~\cite{pass_probchannel}, corresponding to between 2.5 and 10 cents over the last 
five years. 
As depicted in \autoref{fig:motivation}(a) and \autoref{fig:motivation}(d) the average transaction fees of Bitcoin decreased
to around \$0.5 in recent six months while staying around \$0.1 in Ethereum. 
Meanwhile, \autoref{fig:motivation}(b) shows a vivid insight that most transactions fees in
Bitcoin stand between \$0.025 and \$0.125. 

The concept of watchtowers may be economically viable for most generic payment cases which 
can achieve quick channel termination.
However, customers in such scenarios like online gaming, 
video streaming, or energy markets~\cite{probability}, should perform transactions 
almost continuously with tiny amounts.
In \autoref{fig:motivation}(c) and \autoref{fig:motivation}(e), we show that for
the Ethereum platform, the
median transaction value transferred 
in recent half-year is less than \$0.15. Therefore, small payments are still
dominating the majority of all transactions
in Ethereum.

Although, it is hard to estimate the transaction distribution in payment
channels (as this data is not public) we can assure a lower bound for a
watchtower profitability.
We set $\Gamma$ to denote the cost of renting a watchtower and $\Sigma tx$
represents the total
expense of a blockchain transaction in a channel. The value of $\Sigma tx$
includes the value transferred and transaction fees incurred.
When parties perform tiny-amount transactions, i.e., when
$$\Sigma tx < \Gamma$$
it means that the cost of the watchtower involved is higher than the cost the contract
termination by itself.
Therefore, employing a watchtower in such a case is economically ineffective. 

To mitigate this limitation, we introduce short-lived assertions that 
allow channel contracts to distinguish out-of-date and relatively fresh payment
states without relying on any other third party.

\subsection{Mechanism}
\label{sec:shortlived}
A generic limitation of existing payment channels is that smart contracts
operating them cannot distinguish whether a given statement is the latest one,
as an on-chain smart contract does not know an off-chain state without being
updated (payment channels exactly want to limit on-chain updates).  
As a consequence, a stale signed assertion can be sent to a contract which without
knowing about its freshness would temporarily accept it as the current one. 
To mitigate this issue, we introduce \textit{short-lived assertions} that allow
a channel contract to distinguish a fresh state from old ones, and possibly
close the channel with a fresh state quicker.

We assume that 
$[PAY_{k-1},PAY_{k+1}]$
are off-chain transactions between Alice and Bob conducted just after arrival
of blocks 
$block_{i-1}$ and $block_{i+1}$ respectively as depicted in \autoref{fig:assertionexample}.
In our construction, each party appends a \textit{block hash value} when
exchanging a newer micropayment state, such that a state is extended by a
block hash value which acts as a freshness information (i.e., it implies that the transaction is
newer than the given block). 
For instance, $PAY_{k-1}$ would be generated as:
$$
\sigma_{PAY_{k-1}} = \texttt{Sign}_{sk_{A, B}}({\texttt{$bal_{A}$}}  \parallel {\texttt{$bal_{B}$}} \parallel {\texttt{idx}} \parallel {\texttt{$H(block_{i-1})$}}),
$$ 
where $H(block_{i-1})$ represents the hash value of the current block. 
 
As our main goal is to eliminate any third party, like a watchtower,
channel contracts should be able to distinguish by themselves a possibly stale
state from relatively fresh states. 
Interestingly, that can be simply realized, as smart contract languages allow
contracts to access recent block hashes, therefore, a channel contract can
verify that the given state is fresh by verifying it against $n$ last recent
block hashes (where $n$ can be a contract-specific parameter). 
If such a state is submitted it implies that it was created recently, thus the
channel contract can be closed faster as a) it is likely that the state is quite
accurate, thus even if one party successfully misbehaves, its benefit will be
relatively small or marginal, b) it is likely that the both parties are still online, thus
any party can mitigate potential misbehavior.
If the channel contract receives a state which is stale (i.e., older than $n$
last blocks), then it triggers a dispute with a long timeout $T$.

Short-lived assertion is a complementary 
solution that can be applicable for circumstances in which watchtowers are economically unjustified. 
We believe that it may be especially interesting in applications like
micro- or nanopayments.
Another interesting property of this solution is that it ensures both parties
that they see the same version of the blockchain (they sign the same recent
block hash with every transaction), thus it can additionally protect them from
some blockchain-fork related attacks.

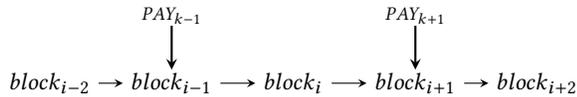
\begin{figure}[h!]
\vspace{-0.3cm} 
\setlength{\abovecaptionskip}{0.1cm}
\setlength{\belowcaptionskip}{-0.3cm}
    \begin{center}
        \begin{tikzpicture}[scale = 0.9,>=stealth]
        \node (A) at (0, 0) {$block_{i-2}$};
        \node (B) at (1.8, 0) {$block_{i-1}$};
        \node (C) at (3.6, 0) {$block_{i}$};
        \node (D) at (5.4, 0) {$block_{i+1}$};
        \node (E) at (7.2, 0) {$block_{i+2}$};


        \draw[->, thick] (1.8, 0.85) -- (1.8, 0.2);
        \node[scale = 0.8] at (1.8, 1.0) {$PAY_{k-1}$};


        \draw[->, thick] (5.4, 0.85) -- (5.4, 0.2);
        \node[scale = 0.8] at (5.4, 1.0) {$PAY_{k+1}$};


        \draw[->]
          (A) edge (B) (B) edge (C) (C) edge (D) (D) edge (E);

        \end{tikzpicture}   
        \caption{Short-lived Assertions Example}
        \label{fig:assertionexample}
    \end{center}
\end{figure}

\subsection{Evaluation}
\label{sec:shortlivedevaluate}
In the design of short-lived assertions, each off-chain
transactions is associated with the current block hash value.
A channel contract customizes a \textit{freshness limit} $n$ to determine whether
the statement is relatively fresh or not via simply checking how old is the
block hash from the statement.
The cost incurred depends on the size of the freshness limit and the statement
submitted. 
The larger the limit is set,
the more gas would be spent as more checks need to proceed in the worst
case.
We measure the total cost of short-lived assertions verification (including 
state parsing, state verification, and block hash value
freshness verification)
for the limit values set to  
$2 \times i$ ($1 \leq i \leq 3$) blocks,
and describe the data in \autoref{table:shortlivedcost}.
Please note that the cost results come from the worst case execution 
for every different freshness limit, which means that the channel contract
should verify the freshness check of the block hash value for all blocks within
limit.
As shown in \autoref{table:shortlivedcost}, the total cost for short-lived assertions
with the limit set to two blocks is \$0.05. With the freshness limit
increasing, more block hash evaluations are processed thus the cost
increases as well.

\begin{table}[!t]
\centering
\caption{Cost for short-lived assertions verification.}
\label{table:shortlivedcost}
\begin{tabular}{L{1.5cm}R{1.1cm}R{1.1cm}R{1.1cm}}
\toprule
\multirow{2}*{Cost} & \multicolumn{3}{c}{Relative freshness limit (blocks)}\\ \cmidrule{2-4}
          & \texttt{2}  &\texttt{4} &\texttt{6} \\ \midrule
Total (gas)     &110178   &163365    &203267    \\
USD             &0.055    &0.082    &0.102          \\ \bottomrule
\end{tabular}
\end{table}

\section{Conclusions}
\label{sec:conclusions}
In this paper, we propose novel fail-safe watchtowers that do not constantly watch
on-chain payment channel contracts but watch off-chain transactions and only
send a single on-chain
message periodically. 
We show that our watchtowers are fail-safe
and efficient, easily handling thousands of payment channels. 
Our design minimizes timeouts required in the common case, as a
watchtower can immediately confirm a channel closure without waiting for a
long timeout.

Moreover, we show that watchtowers in general may be economically ineffective in terms of micropayments
and multiple payment scenarios. 
We introduce short-lived assertions that allow channel contracts to distinguish
fresh and stale assertions and process channel closures accordingly.
Combinations and extensions of these frameworks can be an interesting future
research topic. In particular, we would like to investigate fail-safe
watchtowers publishing current channel states with a randomized probabilistic set
implementations (like Bloom filters) instead of a bitmap like presented.

\section*{Acknowledgments}
We thank the anonymous reviewers for their valuable comments and suggestions.
This research is supported by the Ministry of Education, Singapore, under its
MOE AcRF Tier 2 grant (MOE2018-T2-1-111) and by the SUTD SRG ISTD 2017 128
grant. 

\bibliographystyle{ACM-Reference-Format}
\balance 
\bibliography{paper}

\end{document}